\documentclass[conference]{IEEEtran}
\usepackage{times}
\usepackage{threeparttable}
\usepackage{amsmath,amssymb,amsfonts}
\usepackage{color,xcolor}
\usepackage{pgf,tikz}
\usepackage{graphicx}
\usepackage{subfigure}
\usepackage{flushend}
\usepackage{multirow}
\usepackage[ruled]{algorithm2e}
\usepackage[noend]{algpseudocode}
\usepackage{flushend}
\usepackage{enumerate}
\usepackage{slashbox}

\DeclareMathOperator*{\argmin}{arg\,min}
\DeclareMathOperator*{\argmax}{arg\,max}

\newcommand{\myitem}{{\setlength{\parindent}{8pt}\footnotesize{$\bullet~$ }}}

\newtheorem{lemma}{Lemma}

\begin{document}
\title{Energy Efficient Medium Access with Interference Mitigation in LTE Femtocell Networks}
%\title{Energy Efficiency Enhancement with Interference Mitigation in LTE Femtocell Networks}

\author{\IEEEauthorblockN{Ying Wang, Xiangming Dai, Jason Min Wang, Brahim Bensaou}
\IEEEauthorblockA{Department of Computer Science and Engineering, The Hong Kong University of Science and Technology\\
ywangbf@cse.ust.hk, xdai@cse.ust.hk, jasonwangm@cse.ust.hk, brahim@cse.ust.hk } }

\maketitle

\begin{abstract}

With the rapidly increasing number of deployed LTE femtocell base stations (FBS), energy consumption of femtocell networks has become a serious environmental issue. Therefore, energy-efficient protocols are needed to balance the trade-off between energy saving and bandwidth utilization. The key component of the desired protocol to prevent both energy and bandwidth waste is interference mitigation, which, nevertheless, most previous work has failed to properly consider. To this end, in this paper, we manipulate user equipment (UE) association and OFDMA scheduling with a combination of interference mitigation. Recognizing the NP-hardness of the problem, we propose two distributed algorithms with guaranteed convergence. Extensive simulations show that our algorithms outperform the alternative algorithms in multiple metrics such as utility, power consumption, and convergence speed.
\end{abstract}

% [A simulated annealing based algorithm is also proposed as an optimality benchmark.]\footnote{is it necessary to mention?}

\IEEEpeerreviewmaketitle

\section{Introduction}
%Deploying femtocell base stations (FBS) in LTE networks has recently become highly popular as an effective and economic way to improve data rates and network coverage. According to a recent small cell market status report from the Smallcell Forum \cite{Informa2013}, the number of femtocell base stations (FBS) will grow from 2.5 million in 2012 to 59 million by 2016, a 24 times increase. The huge and rapidly growing number of FBSs has led to significant network energy cost with strong impact on environment and economy, despite of the limited individual FBS power consumption. By 2016, the power consumption of all the FBSs will reach $2.2\times 10^{11}$ KWH, producing several millions tonnes of CO$_2$ per year \footnote{Assume each FBS consumes 10.2 Watt power.}. Based on the principle of green information and communication technology (green ICT), energy-saving designs for femtocell networks are in urge demand.

Femtocell base stations (FBSs) have gained recently a great popularity as an effective and economical approach to increasing the capacity of the LTE macro-cellular networks. In a recent report, the Smallcell Forum, predicts the number of deployed FBSs to reach a sharp 24-fold increase (from 2.5 millions in 2012 to 59 millions by 2016) \cite{Informa2013}. Driven by the need of mobile network operators to cut operational expenditure and the pressure to reduce industrial carbon footprints, in the design of new FBSs hardware, a great care has been given to the energy efficiency. As such, an individual FBS per se consumes, relatively, very little power, however, the consumption of an entire large scale network of FBSs add up to a  huge amounts. By 2016, the power consumption of all FBSs will reach $2.2\times 10^{11}$ KWH, which amounts to producing several millions of tonnes of carbon dioxide per year \cite{Informa2013}.

Reducing power consumption often means switching off as many FBSs as possible \cite{conte2011cell}; yet, doing it in a naive way will inevitably lead to perceptible performance degradation for end users, which goes against the initial rationale of deploying FBSs. Therefore, any feasible energy-saving design for the femtocell networks must bear in mind the quality of service (QoS) of end users as well. In LTE, QoS is implemented between user equipment (UE) and PDN Gateway\footnote{The PDN gateway of packet data network gateway in the LTE jargon is the gateway to wired networks such as the Internet.}, and is characterized by a type of \textsl{bearers} each characterized with some levels of guarantee on some metrics such as delay, loss or throughput. 
%Bearer conceptually is a set of network configuration that determines how the user data is treated when it travels across the LTE network, e.g., prioritizing VoIP traffic over Web browser traffic. 
In femtocell networks, the guarantee of QoS in mainly focused on the wireless link between the UE and the FBS, also known as the ``\textsl{radio bearer}''. Currently two types of bearers exist: guaranteed bit rate (GBR) and non-guaranteed bit rate (Non-GBR). GBR bearers are applied to real-time services such as VoIP and online gaming, while Non-GBR bearers are often applied to elastic traffic such as web browsing and video streaming.

%Although many matured energy-saving techniques have been developed, simply saving the energy is far from enough without considering quality of service (QoS) of the user equipments (UE). In LTE, QoS is defined in the unit of bearers, which is simply an end-to-end flow between the UE and the gateway with predefined QoS requirements. We focus on the QoS of the radio bearer that spans the wireless link between an FBS and a UE. In terms of bandwidth requirement, the bearers can be classified into two main categories: guaranteed bit rate (GBR) bearers for real-time traffic and non-guaranteed bit rate (non-GBR) bearers for elastic traffic. The QoS of both bears should be maintained to ensure the whole bearer’s QoS requirement. Therefore, our task is to design an energy-efficient system to balance the trade-off between energy-saving and QoS satisfaction.

There has been a wide spectrum of models of femtocell networks that target energy efficiency as their primary goal. These models in general fail to take into account the following three aspects together: i) the QoS requirement of the radio bearer; ii) the impact of interference between FBSs; and iii) the gap between the short time scale of orthogonal frequency division multiple access (OFDMA) scheduling and the relatively larger latency for X-2 interface-based coordination between FBSs. For example, FBSs are switched off by manipulating UE association to reduce power consumption in \cite{conte2011cell,zhou2009green,peng2014greenbsn,son2011base}, where the impact of interference is missing. In addition, the QoS requirement is ignored in \cite{conte2011cell,peng2014greenbsn}, while only real-time traffic is considered in \cite{zhou2009green,son2011base}. Both spectrum and energy efficiency are achieved in \cite{hou2013energy} via a joint optimization of resource block scheduling, power allocation, and UE association; however the model is limited to elastic traffic only. Moreover, the distributed protocol proposed in \cite{hou2013energy} imposes severe challenge on the latency requirement of the X-2 interface. Similarly, the centralized mechanism proposed in \cite{saker2012optimal} is inconsistent with the self-organizing nature of femtocell networks. In addition, several designs proposed to selectively activate the sleep-mode of FBSs \cite{claussen2010dynamic}, but the need for hardware modification makes them not immediately practical. In \cite{dufkova2011energy}, the energy efficiency of macrocell, microcell, and femtocell are studied from the perspective of future deployment without giving any concrete solution for existing network.

In this paper, we study the problem of achieving both bandwidth efficiency and energy efficiency in the femtocell network, while taking into account the user perceived quality of service. We propose a generic model that incorporates the key characteristics of the LTE femtocell network: i) using OFDMA; and ii) considering the multi-cell multi-link interference. The gross data rates of UEs represent bandwidth efficiency. The energy efficiency is characterized by the power consumption of both wireless transmission and the operation of base stations. In essence, maximizing the bandwidth utilization is at odds with minimizing the gross energy consumption. We introduce a weight parameter in the model to strike a balance between these two conflicting goals. The fundamental part of our model lies in the joint optimization of UE association and OFDMA scheduling. The UE association determines the sleep/active states of FBSs, which directly influences the power consumption. We also propose a practical MAC protocol on top of OFDMA to allocate each OFDMA tile (the smallest allocation unit in OFDMA) to UEs in a probabilistic manner minimizing thus the communication overhead over the X-2 interface.

%To the best of our knowledge, we are the first to take on all these above challenges. In this paper, we study the problem of balancing the trade-off between energy-saving and QoS satisfaction taking into count interference in LTE femtocell networks with heterogeneous traffic requirement. In our model, we manipulate the sleep/active states of the FBSs by UE association control, and propose a practical MAC protocol with low X-2 interface coordination overhead for effective energy and bandwidth utilization. On recognizing the NP-hardness of the problem, we then propose two iterative algorithms to obtain efficient association and channel access solutions.

%%%%%%%%%%%%%%%%%%%%%%%%[TO-DO]
In summary, this paper's contributions are as follows:

\myitem To the best of our knowledge, this paper is the first to study the problem of energy efficiency  with interference management in LTE femtocell networks with heterogeneous QoS requirements.

\myitem We propose two greedy algorithms that can be implemented in a practical MAC protocol resulting in a low X-2 coordination traffic.

\myitem We evaluate our methods via extensive simulation experiments to demonstrate that our algorithms achieve considerable improvements over some benchmark algorithms from the literature.
%%%%%%%%%%%%%%%%%%%%%%%%[TO-DO]

The rest of the paper is organized as follows. We present the system model in Sec.~\ref{model} and describe two algorithms in Sec.~\ref{algorithm}. We give the performance evaluation in Sec.~\ref{simulation} and conclude the paper in Sec.~\ref{conclution}.

\section{System Model}\label{model}

In this section, we first introduce briefly the femtocell network scenario and clarify the concepts of association control and OFDMA scheduling then finally describe the problem formulation.

\noindent\textbf{Network scenario.} We consider a number of FBSs within a range of a macrocell base station (MBS). The set of FBSs is defined as $B$ and we only consider the downlink transmission. FBSs are randomly deployed by residential users and thus their coverage areas may mutually partially overlap. Let $U$ be the set of all UEs that belong to any owner of FBS. Each FBS is assumed to adopt open-access to all the UEs in $U$. UEs not in $U$ or outside the service range of the FBSs are served directly by the MBS. To avoid the MBS-to-FBS interference, the MBS occupies a different bandwidth frequency from that shared by all FBSs. We partition $U$ into two disjoint sets\footnotemark $U_G$ and $U_N$ to reflect the QoS requirement of GBR bearers and Non-GBR bearers respectively. Each UE $u\in U_G$ is guaranteed to achieve a rate $d_u$, while the rate of a UE $u\in U_N$ is upper bounded by $c_u$ \cite{sesia2009lte}.
\footnotetext{Here, we assume each UE is associated with one bearer. To deal with the case of multiple bearers per UE, the same amount of virtual UEs can be created to fit in our model.}

\noindent\textbf{Association control.} Each UE $u$ can be served by anyone among a set of multiple eligible FBSs. Denote by $N_b$ the set of UEs within the range of FBS $b$. The basic observation to achieving energy efficiency lies in determining the association of UEs, that favours FBSs with low interference and abundant residual bandwidth. FBSs with no UE association will be put into sleep mode. Let $I_{b,u}\in \{0, 1\}$ be the decision variable that indicates whether UE $u$ is associated with FBS $b$. 

\noindent\textbf{OFDMA scheduling.} The interference is coupled with not only the UE association but also the MAC layer resource allocation. We call the smallest allocation unit in OFDMA a ``tile''. Previous approaches \cite{hou2013energy} to scheduling OFDMA tiles rely on the coordination between FBSs over the X-2 interface with milliseconds latency and often incur high computational overhead as well. To achieve low coordination and computational overheads, we propose a probabilistic channel access protocol, which is similar to the one in \cite{hou2011proportionally}. In this protocol, each FBS will allocate each OFDMA tile to its associated UEs probabilistically according to a pre-determined distribution, which shall be returned by our model. Let $p_u$ be the channel access probability of UE $u$ for each title, $0\leq p_u\leq 1$. In fact, we can view this protocol as an enhanced variant of the slotted ALOHA by replacing the time slot as the OFDMA tile and avoiding collisions intra-FBS. Fig.~\ref{fig:1} gives a simple scenario of UE association and the randomized OFDMA scheduling.

\begin{figure}[t]
\centering
\includegraphics[width=0.8\columnwidth]{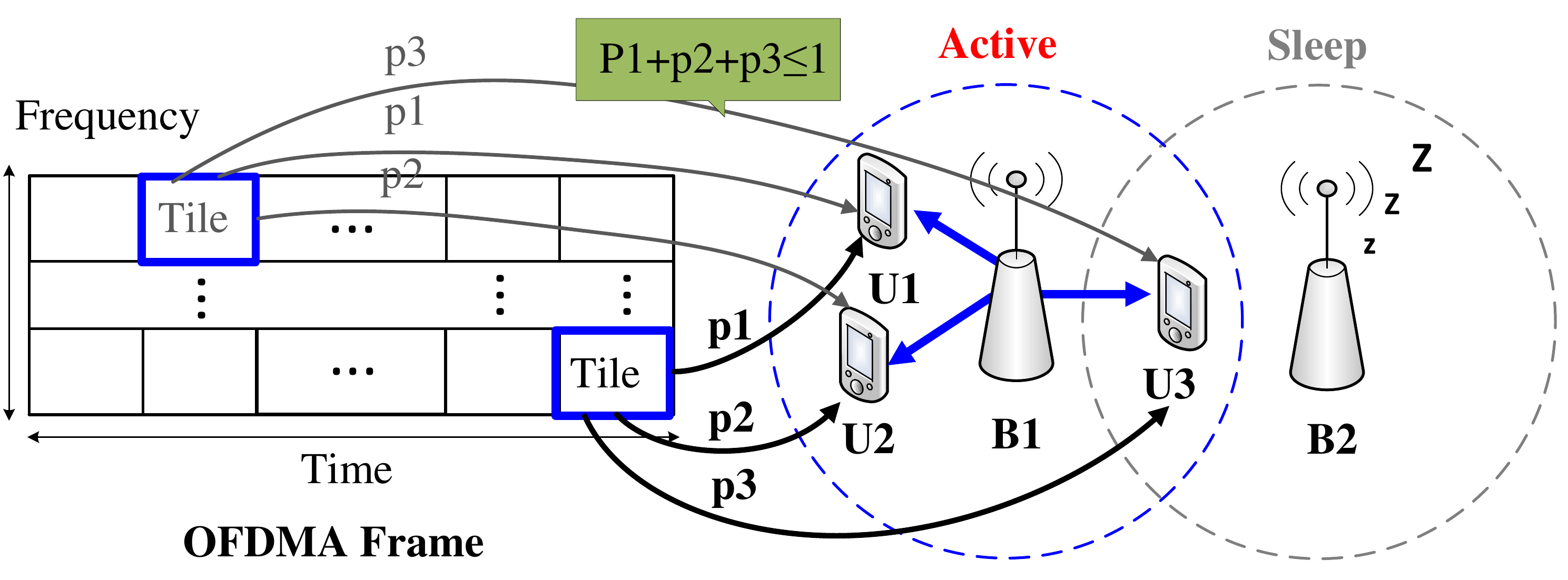}
\caption{Illustration of user association and OFDMA scheduling: (1) all three UEs are associated with B1 and thus B2 will be put into sleep; (2) every OFDMA tile will be scheduled by B1 to the UEs probabilistically based on $\{p_1, p_2, p_3\}$; (3) Intra-FBS, collisions are virtual only.}
\label{fig:1}
\end{figure}

\noindent\textbf{Bandwidth efficiency.} Bandwidth efficiency is characterized by the total throughput of the femtocell network. Let $r_u$ be the data rate achieved by UE $u$. We use utility functions $\mathcal{F}_G (r_u,d_u)$ and $\mathcal{F}_N (r_u,c_u)$ to reflect the satisfaction of a GBR bearer and a Non-GBR bearer UE respectively. $\mathcal{F}_G (r_u,d_u)$ is defined:
\begin{equation}
\mathcal{F}_G (r_u,d_u) = 
\begin{cases}
0,r_u < d_u \\
\mathcal{C}_1, r_u \geq d_u
\end{cases}
\label{eq:1}
\end{equation}
where $d_u$ is the guaranteed data rate of GBR bearer and $\mathcal{C}_1$ is a constant parameter. $\mathcal{F}_G (r_u,d_u)$ is a step function that conveys the strict data rate requirement of GBR bearer. $\mathcal{F}_N (r_u,c_u)$ is defined as:
\begin{equation}
\mathcal{F}_N (r_u,c_u) = 
\begin{cases}
\frac{\mathcal{C}_2 }{log(c_u+1)}log(r_u+1) ,r_u < c_u \\
\mathcal{C}_2, r_u \geq c_u
\end{cases}
\end{equation}
where to comply with the standard, $c_u$ limits the data rate that can be achieved by UEs with Non-GBR bearer and $\mathcal{C}_2$ is also a constant parameter. $\mathcal{F}_N (r_u,c_u)$ implements this in a soft manner, i.e., with zero increase in utility when $r_u\geq c_u$. In addition, $\mathcal{F}_N (r_u,c_u)$ uses $\log(\cdot)$ to apply proportional fairness, which is suitable for elastic traffic. It is worth noting that to enforce higher priorities of GBR-bearer UEs over Non-GBR-bearer UEs, we set $\mathcal{C}_1 > \mathcal{C}_2$.

The throughput of UE $u$, $r_u$, is calculated as:
\begin{equation}
r_u = Rp_u \prod\limits_{
\begin{subarray}{c}
b: u\in N_b,\\
 \ I_{b,u} = 0
\end{subarray}}(1-P_b),	u \in U
\end{equation}
where $P_b = \sum_{u \in N_b} p_u I_{b,u}$ is the gross transmission probability of UEs associated with FBS $b$, $r_u$ is the product of the nominal data rate $R$ and the expected successful transmission probability, which is equal to the product of access probability $p_u$ and the probability that no other contender FBS is transmitting. 

\noindent\textbf{Energy efficiency.} Energy efficiency is characterized by the total power consumption of FBSs. The power consumption of each FBS $E_b$ includes three parts: 1) constant power $E_1$ for transceiver idling in both active and sleeping mode; 2) extra constant power (for computation, backhaul communication, and power supply) $E_2$ in active mode; 3) the transmission power, which is proportional to $P_b$ and the maximum transmission power $E_3$. 
\begin{equation}
 E_b = E_1+ E_2 M_b + E_3 P_b, b\in B
 \end{equation}
where $M_b =  \max_{u\in U}\{I_{b,u}\}$ indicating whether FBS $b$ is in sleeping mode or otherwise (i.e., if at least one UE is associated with FBS $b$ then it must be in active mode).

\noindent\textbf{Problem formulation.} As mentioned before, we use a positive weight $\omega$ to strike a balance between bandwidth efficiency and energy efficiency. Thus the problem is formulated as:
\begin{subequations}
\begin{align}
\max\limits_{I_{b,u}, p_u} & \sum\limits_{u \in U_G} \mathcal{F}_G (r_u,d_u) + \!\sum\limits_{u \in U_N} \mathcal{F}_N (r_u,c_u) - \omega \sum\limits_{b \in B}E_b \label{eq:obj}\\
&  \sum\limits_{b: u\in N_b}I_{b,u} = 1, \forall u\in U \label{eq:c1}\\
& P_b \leq 1, \forall b\in B \label{eq:c2}
\end{align}
\label{eq:origin}
\end{subequations}\\
where constraint (\ref{eq:c1}) enforces the requirement that each UE can be associated with one and only one FBS; and constraint (\ref{eq:c2}) states that for each FBS, each OFDMA tile can be assigned to at most one UE.

% Negative constant $\mathcal{C}_1$ represents the penalty of not satisfying the demand. Smaller the value of $\mathcal{C}_1$ is, more serious requirement the UE has ($\mathcal{C}_1=-\infty$ in extreme case). The non-GBR UEs deal with elastic traffic. 

\noindent\textbf{NP-hardness.} We consider a special case of problem (\ref{eq:origin}) where: i) each FBS provides full data rate $R$ without any inter-cell interference; ii) only UEs with GBR-bearer exist and FBSs are able to serve UEs with full guaranteed rates. In this case, only the energy part is left in the objective function (\ref{eq:obj}). By setting $E_3 \ll E_2$, all active FBSs are assumed to consume the same amount of power. Now, this problem becomes packing UEs with different demands into a minimum number of FBSs that have the capacity $R$, which is a bin-packing problem. Therefore, the NP-hardness of problem (\ref{eq:origin}) is established by the NP-hardness of this special case.

\section{Proposed Algorithms}
\label{algorithm}

%Problem \eqref{origin} is a mixed-integer problem and usually difficult to solve. Although we can relax the integer variables, the non-concavity of the objective and constraints still makes it impossible to devise even approximation methods via traditional nonlinear optimization techniques. Moreover, a distributed algorithm is needed for FBSs without central control, which makes the problem even harder.

In the context of femtocell networks, we seek feasible efficient distributed algorithms to solve problem (\ref{eq:origin}) approximately. We first discretize the continuous decision variable $\{p_u\}$ in order to reduce the solution search space. The discretized variable of $p_u$ is represented as $q_u$ where $q_u \in \mathcal{Q} = \{ 0, \frac{1}{n-1},\ldots, \frac{n-2}{n-1}, 1\}$. Let $s_u$ be the decision vector for each UE $u$: $s_u =(I_{1,u},...,I_{b,u},...,I_{|B|,u}, q_u )$. The feasible set of $s_u$ is defined as $\mathcal{S}_u = \{s_u|\sum_{b: u\in N_b}I_{b,u} = 1, I_{b,u} = \{ 0,1\}, q_u \in \mathcal{Q} \}$. Note that the constraint (\ref{eq:c1}) has been incorporated into $\mathcal{S}_u$. Then, a discretized variant of problem (\ref{eq:origin}) can be written as:
\begin{equation}
\begin{array}{ll}
\max\limits_{\{s_u\}} & \sum\limits_{u \in U_G} \mathcal{F}_G (r_u,d_u) + \sum\limits_{u \in U_N} \mathcal{F}_N (r_u,c_u) - \omega \sum\limits_{b \in B}E_b \\
& + \sum\limits_{b \in B}\mathcal{G}(P_b-1)\\
\mbox{s.t. } & s_u \in \mathcal{S}_u, \forall u\in U
\end{array}
\label{eq:newobj}
\end{equation}\\
where the function $\mathcal{G}(x)$ is defined as
\begin{equation}
\mathcal{G}(x) = 
\begin{cases}
0,x \leq 0 \\
-\mathcal{C}_3x, x > 0.
\end{cases}
\end{equation}\\
$\mathcal{C}_3$ is a very large positive constant, as a penalty for violating constraint (\ref{eq:c2}). Note that the term $\mathcal{G}(P_b - 1)$ in the objective function of (\ref{eq:newobj}) implicitly enforces constraint (\ref{eq:c2}) in the original problem (\ref{eq:origin}).

%To further simply the problem, 
%we remove the constraints specified in (\ref{transmission}) by adopting distance based static penalty function $\eta(\cdot)$ \cite{smith1997penalty}, where constant $\mathcal{C}_3$ is a negative penalty ratio imposed for violation of the constraint ($\mathcal{C}_3= -\infty$ in extreme case):
%\begin{equation}
%\eta(x) = 
%\begin{cases}
%0,x \leq 0 \\
%\mathcal{C}_3x, x > 0,
%\end{cases}
%\end{equation}

%Define for UE $u$ an association vector $(I_{1,u},...,I_{b,u},...,I_{|B|,u})$, which is selected from a finite vector set $\mathcal{I}_u = \{(I_{1,u},...,I_{b,u},...,I_{|B|,u})|\sum_{b: u\in N_b}I_{b,u} = 1, I_{b,u} = \{ 0,1\} \}$.  Then the decision vector of each UE is defined as $s_u = (I_{1,u},...,I_{b,u},...,I_{|B|,u}, q_u )$. Each UE's decision set is defined as $\mathcal{S}_u = \mathcal{Q}_u \times \mathcal{I}_u$.

%(The aggregation of penalty functions, is strictly monotonic. Any increase of the number of satisfied constraints will result in an increase in the new objective.)

%The new discretized and unconstrained problem now becomes:
%\begin{equation}
%\begin{split} 
%\max\limits_{s_u \in \mathcal{S}_u} & \sum\limits_{u \in U_G} \mathcal{F}_G (r_u,d_u) + \sum\limits_{u \in U_N} \mathcal{F}_N (r_u,c_u) - \omega \sum\limits_{b \in B}E_b \\
%& + \sum\limits_{b \in B}\eta(P_b-1) .
%\label{eq:newobj}
%\end{split} 
%\end{equation}

Next, we propose two distributed algorithms to obtain approximate solutions to the transformed problem (\ref{eq:newobj}): \textit{iterative greedy algorithm} and \textit{fast iterative greedy algorithm}.

\subsection{Iterative Greedy (IG) Algorithm}
%(wang ying) Each UE generates a UE index for each iteration according to a predefined random UE index sequence (all UEs share the same sequence), such that in each iteration only one UE index is selected out. By comparing the generated UE index with its own index, each UE knows whether to update its decision in the current iteration.

%The general idea of the \textit{iterative greedy algorithm} are as follows: in each iteration, a UE index is randomly generated by the UEs locally\footnote{This was achieved by maintaining a random generator at each UE with the same seed.}. Then the chosen UE updates its decision by greedily selecting a decision that maximizes its local objective function\footnote{If multiple decisions achieve the same maximum local objective, each decision is selected with equal probability.}, assuming other UEs keep the same decision as the last iteration. 

The IG algorithm proceeds in an iterative manner. In each iteration, one UE is selected to make a local optimal decision, i.e., maximizing its local profit $\theta_u$, while states of other UEs remain the same as the last iteration. Let $\alpha_u = \{\cup_{b}N_b | u\in N_b\}$ (where $u\in \alpha_u$). Clearly, the local decision of $u$ can influence utilities of other UEs in $\alpha_u$ and the power consumption of FBSs in $N_b$. We incorporate these concerns into $\theta_u$ so that maximizing local profit $\theta_u$ will always lead to an increase of the global profit (i.e., the objective function of (\ref{eq:newobj})), which shall be proved in Lemma~\ref{lemma1}. Deliberately, $\theta_u(s_u)$ is defined as 
\begin{equation}
\begin{split}
\theta_u(s_u) = & \sum\limits_{u\in \alpha_u \cap U_G} \mathcal{F}_G (r_u,d_u) + \sum\limits_{u\in \alpha_u \cap U_N} \mathcal{F}_N (r_u,c_u)\\ & -\omega \sum\limits_{b:u\in N_b} E_b 
+ \sum\limits_{b: u\in N_b}\mathcal{G}(P_b-1),
\end{split}
\label{theta}
\end{equation}
which includes three parts: i) the sum of utilities of neighbour UEs $\alpha_u$; ii) the sum of power consumption of FBSs that are eligible to serve $u$ (i.e., $u\in N_b$); and iii) the penalty part to enforce constraints (\ref{eq:c2}) coupled with $u$. In order to calculate the power and penalty parts, the decisions of UEs in $\alpha_u\backslash \{u\}$, i.e., $\mathbf{s}_{\alpha_u\backslash \{u\}}$, must be known. However, since each UE's utility is influenced by the decisions of its neighbour UEs, to calculate the utilities of UEs in $u^{\prime}\in \alpha_u\backslash \{u\}$, we must know the decisions of UEs $u^{*}\in \alpha_{u^{\prime}}$ as well, i.e., $s_{u^{*}}$. For notational convenience, we define $\alpha^{(2)}_u = \{\cup \alpha_{u^{\prime}}| {u^{\prime} \in \alpha_u} \}$ (where $u\in \alpha^{(2)}_u$) as the two-tier neighbour set of UE $u$. Then, in $t$-th iteration, maximizing the local profit $\theta_u$ amounts to finding ${s_u^{(t)}} = \argmax_{s_u \in \mathcal{S}_u}\theta_u(s_u, \boldsymbol{s}_{\alpha^{(2)}_u \backslash \{u\}}^{(t-1)})$, where $\boldsymbol{s}_{\alpha^{(2)}_u \backslash \{u\}}^{(t-1)}$ is the union of decisions of UEs $\alpha^{(2)}_u \backslash \{u\}$ in $(t-1)$-th iteration.

\begin{algorithm}[h]
\caption{Iterative Greedy Algorithm}
\label{IG}
Set initial decision $s_u^{(1)} = \boldsymbol{0}$; $t \gets 1$;\\
\Repeat{Convergence}{
Obtain $\boldsymbol{s}_{\alpha^{(2)}_u \backslash \{u\}}^{(t-1)}$ of set $\alpha^{(2)}_u\slash\{u\}$ in iteration $(t-1)$;\\
\If{$u$ is selected}{	
	Update decision: $s_u^{t} = \argmax\limits_{s_u \in \mathcal{S}_u}\theta_u(s_u,\boldsymbol{s}_{\alpha^{(2)}_u \backslash \{u\}}^{(t-1)})$;\\
	Notify decision $s_u^{t}$ to UEs in $\alpha^{(2)}_u\slash\{u\}$;\\
}
$t\gets t+1$;
}	
\end{algorithm}

The IG algorithm is summarized in Alg.~\ref{IG}. All UEs are allocated with initial decisions in the first iteration. The complexity of IG algorithm in each iteration is $O(n|B|)$, where $n$ is the size of the discrete set $\mathcal{Q}$, and $|B|$ is the number of FBSs.

Alg.~\ref{IG} is an iterative algorithm and we prove its convergence via Lemma~\ref{lemma1}.
\begin{lemma} 
Alg.~\ref{IG} generates an improvement path of a potential game and thus converges to an equilibrium in a finite number of steps.
\label{lemma1}
\end{lemma}

\begin{proof}
We define a strategic game $G=(U,\langle \mathcal{S}_u \rangle,\langle\theta_u\rangle)$, where each UE $u$ is a player with strategy $s_u \in \mathcal{S}_u$, and its payoff function is $\theta_u$. Next, we are going to show that the objective function in (\ref{eq:newobj}), denoted as $\mathcal{P}$, is a potential function. For any UE $u\in U$ and for any two strategies $s^{\prime}_u, s_u\in \mathcal{S}_u$, the following property holds obviously,
\begin{equation}
\label{potential}
\begin{split}
& \mathcal{P}(s^{\prime}_u,\boldsymbol{s}_{U\backslash\{u\}})- \mathcal{P}(s_u,\boldsymbol{s}_{U\backslash\{u\}}) = \\
& \theta_u(s^{\prime}_u,\boldsymbol{s}_{\alpha_u^{(2)}\backslash\{u\}}) - \theta_u(s_u,\boldsymbol{s}_{\alpha_u^{(2)}\backslash\{u\}}), \forall s_u \in \mathcal{S}_u, \forall s^{\prime}_u\in \mathcal{S}_u
\end{split}
\end{equation}\\
With the property \eqref{potential}, game $G$ is an exact potential game \cite{monderer1996potential}. In each iteration of Alg.~\ref{IG}, only one UE $u$ is allowed to update its decision $s_u^*$, which maximizes its payoff $\theta_u$. Therefore, $\theta_u(s_u^*, \boldsymbol{s}_{\alpha_u^{(2)}\backslash\{u\}}) \geq \theta_u(s_u^{\prime}, \boldsymbol{s}_{\alpha_u^{(2)}\backslash\{u\}}), \forall s_u^{\prime}\in \mathcal{S}_u$. Let $ (\boldsymbol{s}_1,\boldsymbol{s}_2,\boldsymbol{s}_3,\cdots)$ be any decision path generated by Alg.~\ref{IG}. According to (\ref{potential}), $\mathcal{P}(\boldsymbol{s}_1) \leq \mathcal{P}(\boldsymbol{s}_2) \leq \mathcal{P}(\boldsymbol{s}_3)\leq\cdots$. That is, there is always improvement on the potential function along this decision path. Since the game has a finite strategy space and each steps improves the potential, Alg.~\ref{IG} will converge in finite steps.
\end{proof}

\subsection{Fast Iterative Greedy (FIG) Algorithm }

One shortcoming of IG algorithm is that only one UE is allowed to update its decision in each iteration and thus Alg.~\ref{IG} may take a long time to converge, impairing the network efficiency. If two UEs $u, u^{\prime}$ are not 2-tier neighbours, i.e., $u^{\prime}\notin \alpha_u^{(2)}$, then the decision of $u^{\prime}$ will not impact $\theta_u$ and vice versa. As a consequence, it is beneficial to allow $u$ and $u^{\prime}$ to update their decisions simultaneously. Based on this observation, we propose a fast iterative greedy algorithm (FIG), which allows as many UEs as possible to update their local decisions and thus will converge faster. To this end, we invoke a distributed colouring algorithm, shown in Alg.~\ref{DC}, which greedily partitions UEs into maximal independent sets \cite{barenboim2013distributed}. The colouring procedure runs over the UE graph, where each vertex corresponds to a UE and an edge exists between two UEs if they are 2-tier neighbours to each other. Initially, each UE is coloured with its UE index and thus $|U|$ colors are used. Then in each iteration, every UE obtains the colors of its 2-tier neighbouring UEs. The UEs assigned with the largest color index seek to reduce their color index by assigning themselves with the minimum possible color indices, while ensuring the assigned colour does not collide with the colours of the neighbours. The colouring algorithm terminates when no UEs can reduce its color index further. The complexity of Alg.~\ref{DC} is $O(\Delta^2|U|)$, where $\Delta$ is the maximum degree of the UE graph and $|U|$ is the number of UEs.

\begin{algorithm}[h]
\caption{Distributed Coloring Algorithm}
\label{DC}
$u.color = u.index$; $t\gets |U|$;\\
\While{$t>0$}{
	Obtain colors of set $\alpha^{(2)}_u\slash\{u\}$ in $(t-1)$th iteration; \\
	\If{$ u.color == t $}{
		%Find the minimum possible color $\kappa$ for $u$ without conflicting with the colors of $\alpha_u$\\
		$k^*= \argmin \{k\neq u^{\prime}.color,k=1,...,|U|,u^{\prime}\in \alpha^{(2)}_u\slash\{u\}$\};\\
		\If{$ u.color > k^* $}{
			$u.color = k^*$;\\
			Notify $u.color$ to UEs in $\alpha^{(2)}_u\slash\{u\}$;
		}
		%\Else{\textbf{break}}
		\Else{\Return $u.color$;}
		}
		$t\gets t - 1$;
}
\end{algorithm}

\begin{lemma} 
The distributed coloring algorithm can reduce the color number to at most $(\Delta+1)$ where $\Delta$ is the maximum degree of the UE graph.
\end{lemma}
\begin{proof}
Suppose the colouring algorithm ends with $(\Delta+2)$ colors in the last iteration. A UE coloured with $(\Delta+2)$ can always find a not-used color from color 1 to $(\Delta + 1)$ since it has at most $\Delta$ neighbours. Therefore, we can always find a $(\Delta+1)$-colouring to any graph with a maximum degree $\Delta$.
\end{proof}

%The fast iterative greedy algorithm is executed after the coloring algorithm. The FIG algorithm is similar to Alg. \ref{IG} except that in the first step, each UE generates a random color index instead of the UE index. The UEs colored by the generated index can update decisions simultaneously. The complexity of the FIG algorithm in each iteration is also $O(n|B|)$.

%In summary, IG and FIG algorithms are able to seek effective local UE decisions in view of global profit maximization with convergence guaranteed. It is also worth mentioning that the solutions obtained from these two algorithms may trap in the local optimal instead of global optimal due to nonconcavity.
FIG proceeds in a similar way as IG except that in each iteration, multiple UEs with the same color will make local decisions simultaneously, whereas in each iteration of IG only one UE can make a local decision. Both FIG and IG will converge to an equilibrium state within limited number of iterations. That is to say, FIG and IG do not yield a global optimum.

\subsection{Implementation Discussion}
%In our algorithms, UEs update local decision on association and access probability by 2-tier neighbourhood coordination.
%However, direct communication between UEs is unrealistic. In practice, it is the FBSs that make decisions for the UEs via inter-FBS X-2 coordination. Since each FBS is aware of UEs in its range, FBSs can make sure every UE's decision is made by one FBS that covers it via X-2 coordination. When making decisions for UEs, FBSs obtain 2-tier neighbourhood information via X-2 interfaces within two-hop transmission.

In each iteration of IG algorithm, UE $u$ relies on the information of its two-tier neighbours $\boldsymbol{s}_{\alpha_u^{(2)}}$ to make a local optimal decision. However, with no direct communication between UEs, the neighbourhood information ($\alpha_u, \alpha_u^{(2)}$) between UEs is unknown, let alone the ``strategy'' information $\boldsymbol{s}_{\alpha_u^{(2)}}$. In fact, in the real implementation, the FBS can play the role of making local decisions $\{s_u^{(t)}\}$ (in Alg.~\ref{IG}) on behalf of UEs $u\in N_b$. Since one UE $u$ may locate in the range of multiple FBSs, only one dedicated FBS needs to be selected (via certain distributed consensus mechanism at association) for each UE. And the needed information $\boldsymbol{s}_{\alpha_u^{(2)}}$ will be easily obtained via coordination over the X-2 interface between FBSs.

%In each iteration, each UE needs the decisions of its 2-tier neighbourhood in the previous iteration to update its decision. When executing the algorithm, Each FBS keeps the decision information of the UEs it makes decision for and also the neighbours of them. Each time when the UE's decision is updated, the responsible FBS updates its information to all the FBSs responsible for its neighours UEs. 

%In our design, the UEs act as independent decision making units, yet work together to maximize the same objective in \eqref{newobj}. Since the FBSs have high computation ability and are aware of the neighbouring traffic and contention information by FBS coordinations, each UE's decision is actually made by one of the nearby FBSs. As FBSs make decisions for the UEs, the neighbourhood information can be obtained by X-2 interfaces between the FBSs. The communication overhead is tolerable with at most two-hop transmission.

\section{Performance Evaluation}\label{simulation}

\subsection{Benchmark Algorithms}

For comparison purpose, we study the performance of a distributed load-aware (LA) algorithm proposed in \cite{zhou2009green}. To the best of our knowledge, the LA algorithm is the closest to our work except that it assumes interference is solved in advanced by other techniques. To ensure fair comparison, the interference mitigation method for OFDMA random access in \cite{hou2011proportionally} (similar as the one in our model) was added on top of the original LA algorithm, with differentiation of priorities of GBR and Non-GBR UEs.

We also propose to use a simulated annealing (SA) algorithm as an upper bound benchmark \cite{borst2014nonconcave}. In SA, each UE updates its decision variables iteratively according to a probability distribution calculated based on a local function given in \eqref{theta}. The SA algorithm converges to global optimal probabilistically when time goes to infinity.

\subsection{Simulation Setup}

We evaluate the performance of the IG algorithm, the FIG algorithm, the SA algorithm, and the LA algorithm using numerical simulation. The simulation setup is as follows.

We focus on downlink transmissions. The transmission range of each FBS is set to 10 meters. Each FBS provides a full rate of 100 Mbps when no collision happens. There are in total $10$ levels of discrete channel access probabilities ($n=10$). For each GBR UE, the rate requirement $d_u$ equals $10$ Mbps, and for each Non-GBR UE, the maximum rate limitation $c_u$ equals 20 Mbps. Each GBR UE's maximum utility $\mathcal{C}_1$ is set as $100$, and each Non-GBR UE's maximum utility $\mathcal{C}_2$ is set as $10$, ensuring thus a higher priorities for GBR UEs over Non-GBR UEs. Constant $\mathcal{C}_3$ is set to $100$. According to the FBS hardware power consumption in \cite{ashraf2011sleep}, the constant power $E_1$ is set to $0.7$ Watt (20\% of the transceiver power), the extra constant power $E_2$ is set to $6.7$ Watt, and the maximum proportional transmission power $E_3$ is set to $2.7$ Watt. Then, the power consumption of an active FBS can reach at most $10.2$ Watt. 
%For SA algorithm, parameter $T_0$ is set as $10$. For LA algorithm, $\epsilon$ is set as 0.01 and parameter $\alpha$ is set as $0$. 

\subsection{Convergence Analysis}

We first examine the algorithm convergence property and performance on a simple topology with three FBSs and eleven UEs (numbered from 1 to 11). As shown in Fig. \ref{ThreeTopology}, UE 1, 3, 6, 9, and 10 are GBR UEs and the others are Non-GBR UEs. The power weight $\omega$ is set to $1.0$ for algorithm IG, FIG, and SA. 

\begin{figure}[!htbp]
\centering
\subfigure[Simple Topology]
{\label{ThreeTopology}\includegraphics[width=0.5\columnwidth]{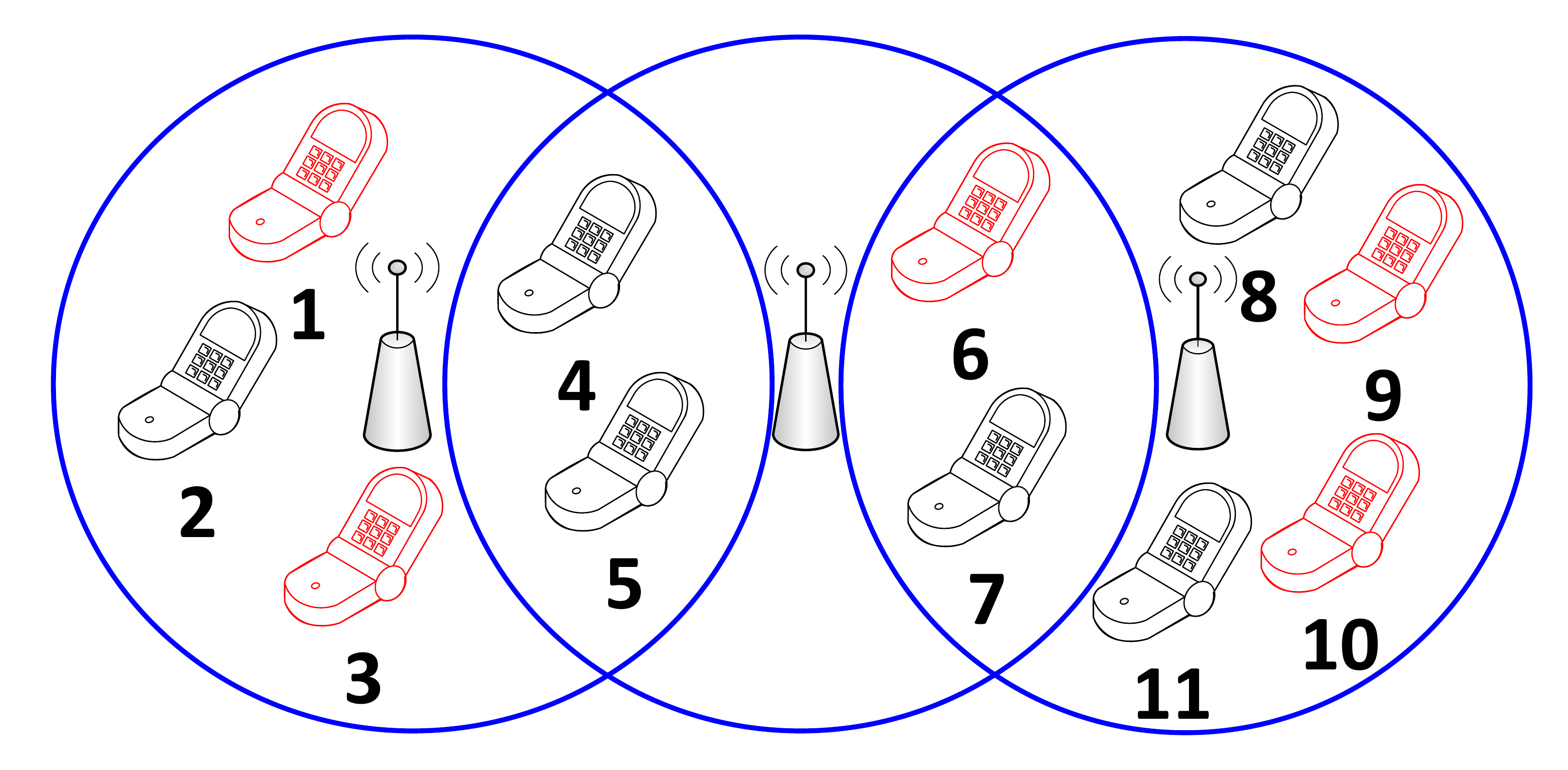}}
\subfigure[Large Topology]
{\label{LargeTopology}\includegraphics[width=0.40\columnwidth]{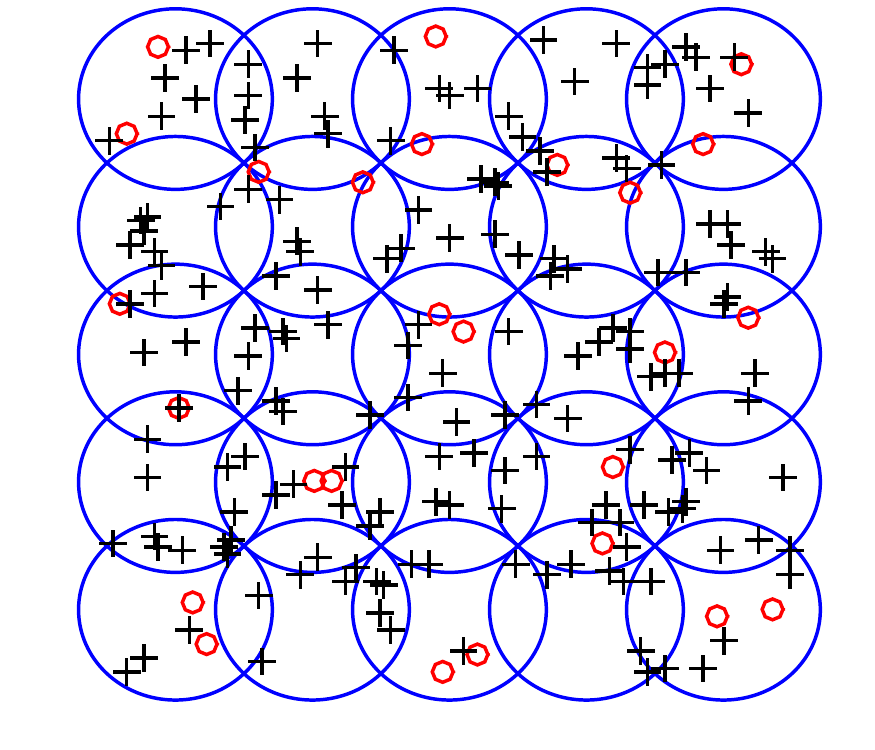}}
\caption{Topologies (The red nodes are GBR UEs, and the black nodes are Non-GBR UEs)}
\end{figure}

We investigate the convergence of three performance metrics: utility, power consumption, and energy efficiency, where utility is the sum of utilities of both GBR and Non-GBR UEs, power consumption is the total power consumption of the network, and energy efficiency is the total UE utility divided by the total network power consumption (or unit of utility per unit of power consumed).  
% Higher the utility value is, higher the QoS satisfaction level the network can provide for its UEs. Network with lower power consumption is obviously an energy-saving network.

Figure \ref{Convergence} shows the convergence curves of the above three metrics. The samples are averaged over 100 independent simulation runs with the same initial UE decision (no association and zero access probability). 

\begin{figure}[!htbp]
\centering
\subfigure[Utility]{\label{SimpleUtility}\includegraphics[width=0.48\columnwidth]{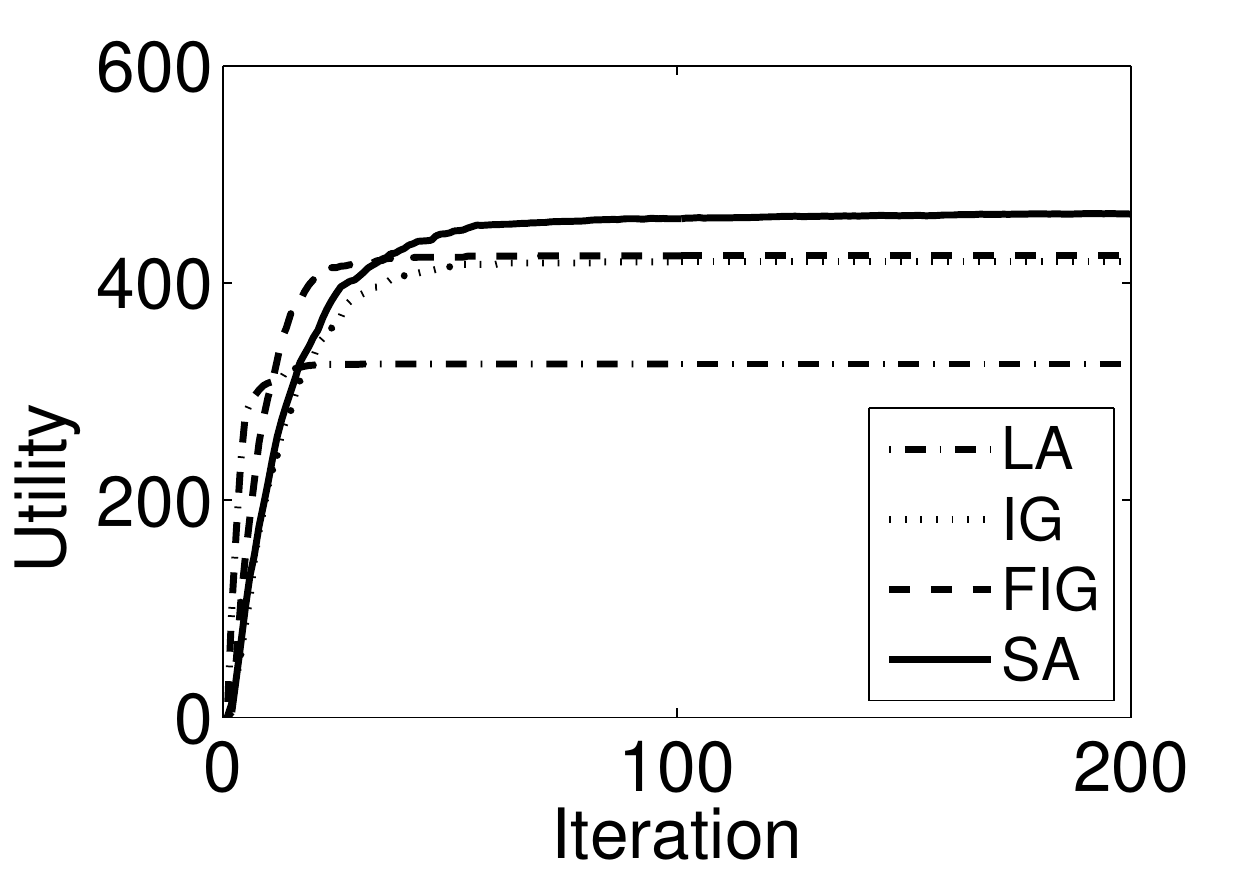}}
\subfigure[Power Consumption]{\label{SimplePower}\includegraphics[width=0.48\columnwidth]{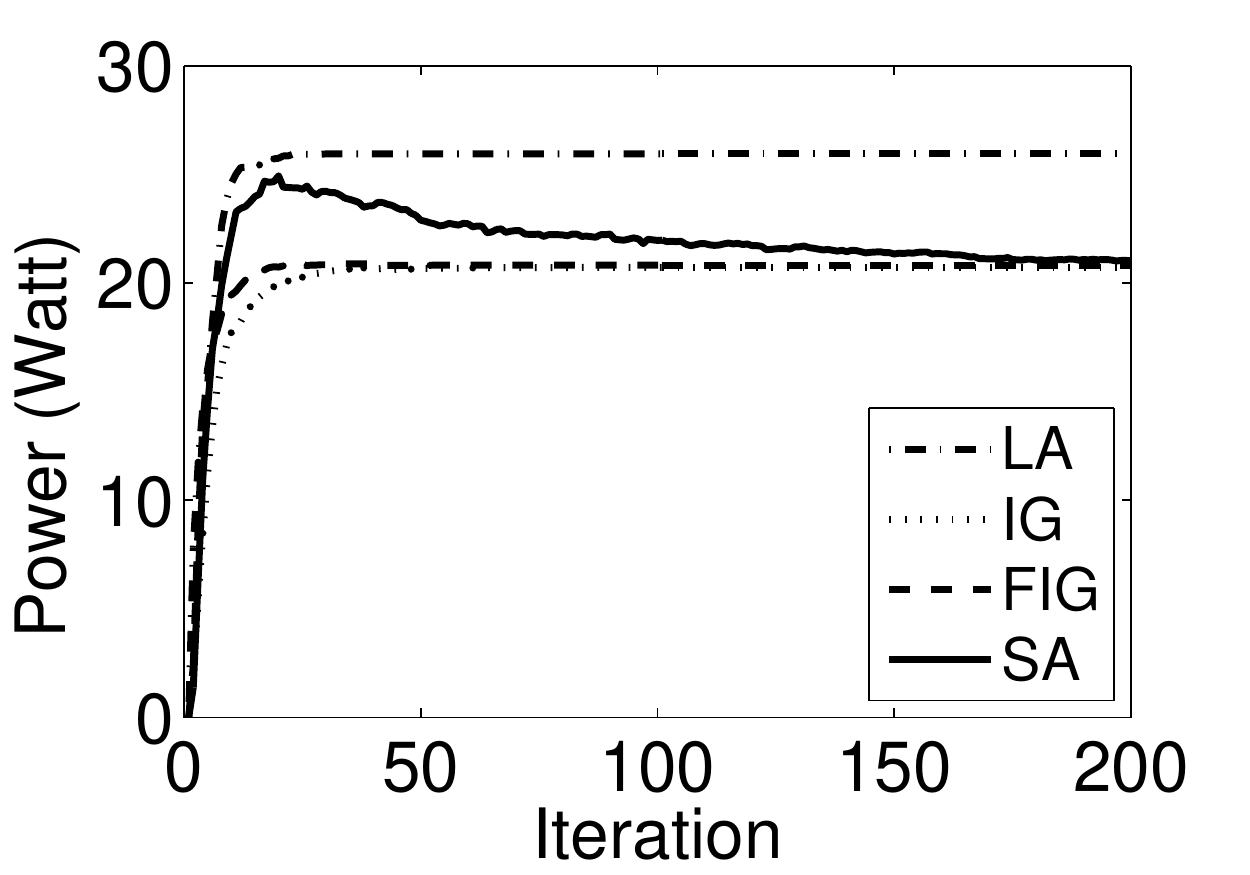}}
\subfigure[Energy Efficiency]{\label{SimpleEfficiency}\includegraphics[width=0.48\columnwidth]{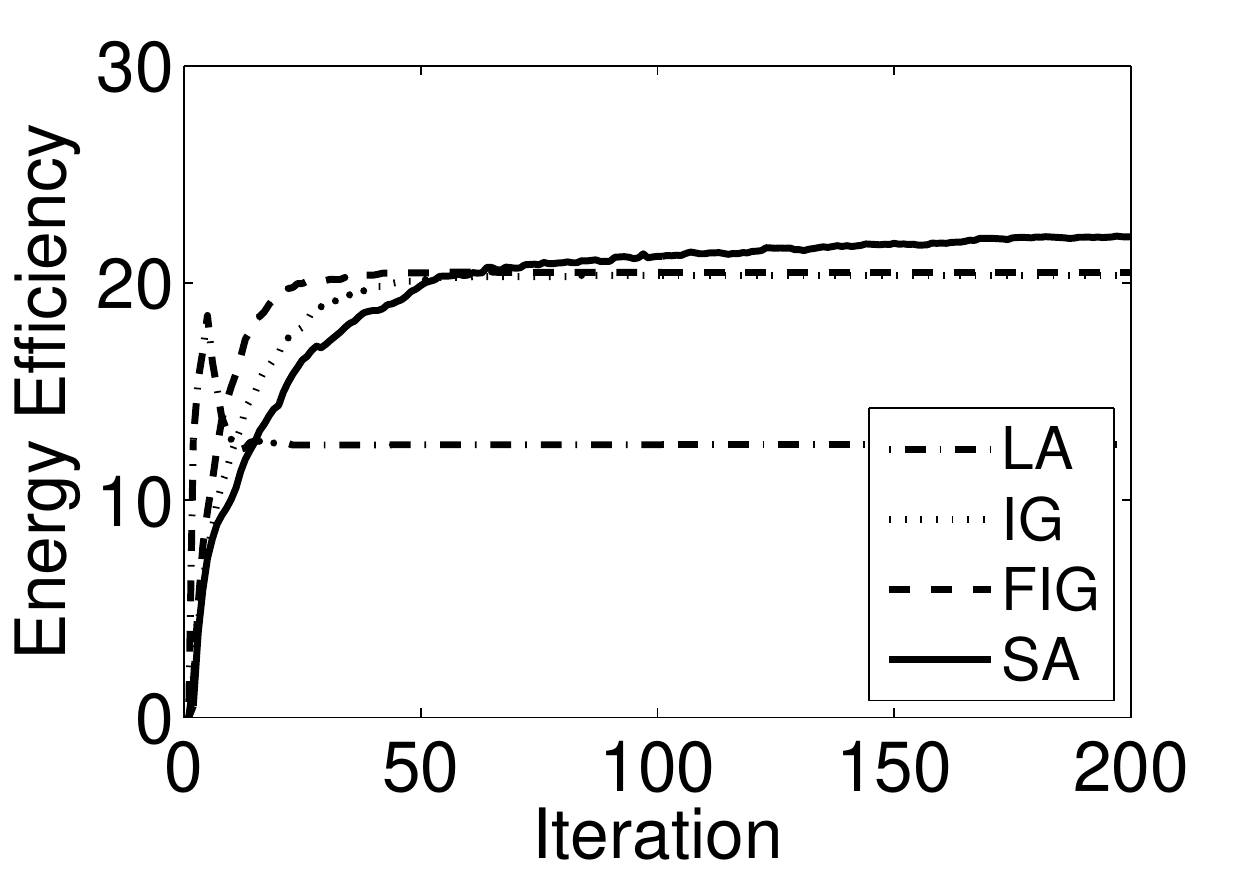}}

%\subfigure[Objective]{\label{SimpleObjective}\includegraphics[width=0.49\columnwidth]{figures/Results/Simple/RealPayoff.pdf}}
\caption{Convergence Analysis}
\label{Convergence}
\end{figure}

\noindent \myitem \textbf{Utility.} The LA algorithm had the highest convergence speed and converged in about 20 iterations as shown in Fig. \ref{SimpleUtility}. Our FIG algorithm converged in less than 30 iterations, which was half the time taken by IG algorithm. The SA algorithm converged slower as it probes an optimal solution using a probabilistic approach. It reached a relatively stable utility in around 150 iterations and kept improving slowly afterwards. 

\noindent \myitem \textbf{Power consumption.} As shown in Fig. \ref{SimplePower}, SA algorithm increased in the first 20 iterations and then drops gradually since the relative benefit of power saving becomes greater than the utility increasing. IG, FIG, and SA algorithms converged at around 20 Watt, indicating that on average only two FBSs were turned on to avoid interference and save energy.

\noindent \myitem \textbf{Energy efficiency.} As shown in Fig. \ref{SimpleEfficiency}, LA algorithm increased abruptly and then followed by a sudden drop, which was caused by the fast increasing power consumption.

Although LA algorithm converged very fast, it received the lowest utility, highest power consumption, and lowest energy efficiency after convergence. IG and FIG algorithms obtained similarly good performance on the three metrics: 30\% higher utility than LA, 20\% less power consumption than LA, and 63\% higher energy efficiency than LA. Compared to the upper bound SA, IG and FIG achieved 8\% lower utility, 7\% lower power consumption, and 6\% lower energy efficiency by the end of the 200th iteration.

In summary, in the simple topology, both IG and FIG algorithms outperform the LA algorithm in utility, power consumption, and energy efficiency. IG and FIG algorithms converged much faster than SA algorithm, in the meanwhile achieved approximately as good performance as SA in all three metrics.

\subsection{Large Scale Topologies}
\begin{table*} \renewcommand{\arraystretch}{1}
\small
\begin{center}
\caption{Performance Metrics Comparison in Large Scale Topology} 
\begin{threeparttable}
	\centering 
	\begin{tabular}{|c|c|c|c|c|c|c|} 
	\hline  
	{\backslashbox{Metric}{Algorithm}  }&{LA}&{IG}&{FIG}&{SA}\\ 
	\hline
	Utility &1263.80$\pm$150.53 & 3152.50$\pm$151.32 & 3171.30$\pm$171.58 & 3456.80$\pm$84.85 \\
	\hline 
	GBR Reject Ratio &59.31\%$\pm$6.07\% & 9.00\%$\pm$5.39\% & 8.42\%$\pm$5.75\% & 2.68\%$\pm$2.87\% \\
	\hline	
	Non-GBR Utility  &205.77$\pm$15.41 & 786.51$\pm$59.14 & 790.29$\pm$65.31 & 926.43$\pm$44.38 \\
	\hline
	Power Consumption (/Watt) &203.76$\pm$0.88 & 213.12$\pm$11.71 & 217.9$\pm$13.33 & 231.94$\pm$7.95 \\
	\hline
	Energy Efficiency &6.20$\pm$0.73 & 14.82$\pm$0.84 & 14.59$\pm$0.85 & 14.92$\pm$0.48 \\
	\hline 
	\end{tabular}
%	\begin{tablenotes}
%       \item[1] The table entries are accurate to two decimal places.
%      % \vspace{-1.5em}
%	\end{tablenotes}
\end{threeparttable}
\label{table_Large} 
\end{center}
\end{table*}

Next, we consider a large scale topology consisting of twenty-five FBSs placed on a five-by-five grid, as shown in Figure \ref{LargeTopology}. In this scenario, each FBS was located at the center of a circle, which represents the transmission range of a FBS. The FBSs overlapped with one another and thereby led to inter-femtocell interference. In each circle, eight UEs were randomly distributed in the circle's inscribed square. Each black cross represents a Non-GBR UE, and each small red circle represents a GBR UE. Two GBR UEs and six Non-GBR UEs were located randomly within every other circle.

Table \ref{table_Large} to Table \ref{table:weight_efficiency} show the averaged results of 100 random UE distributions. For SA algorithm, the metric values were averaged over 100 iterations due to oscillation, while others were captured after convergence. Standard deviation is also reported after the average value in each cell. The table entries are accurate to two decimal places.

\noindent \textbf{Metric Comparison.} Table \ref{table_Large} shows the comparison results on five metrics: utility, power consumption, energy efficiency, GBR reject ratio, and Non-GBR utility. The latter two metrics were designed to evaluate separate profit of GBR and Non-GBR UEs: the GBR reject ratio is defined as the percentage of unsatisfied GBR UEs that should be rejected by the FBSs, while the Non-GBR utility is defined as the sum of utilities of the Non-GBR UEs.

From Table \ref{table_Large}, we observe that in terms of utility, 1) LA algorithm achieved the lowest utility because rough inter-FBS interference mitigation when all FBSs are turned on results in unnecessary bandwidth waste. 2) IG and FIG obtained almost the same utilities as they adopted the same greedy decision updating rules. 3) Both IG and FIG improved the utility value by 150\% over the benchmark LA thanks to the high bandwidth efficiency obtained by intelligent interference mitigation. 4) The utilities of IG and FIG are close to the upper bound benchmark SA (only 10\% lesser).
% and sometimes the result of IG and FIG even overlap with SA result.

On other metrics, our IG and FIG algorithms outperformed LA algorithm and obtained similar good results as SA algorithm: 1) The GBR reject ratio of SA is the lowest (3\%). IG and FIG achieved only less than 9\%, which is only ${1}/{6}$ of LA. 2) The Non-GBR utilities of IG and FIG are nearly three times higher compared to LA, with only 14\% drop compared to the SA. 3) IG and FIG consumed more power than LA, and less power than SA. 4) IG and FIG improved the energy efficiency by more than two times compared to LA. IG and FIG have similar energy efficiencies as SA, as they achieved less utility and consumed less power compared to SA. 5) The standard deviations of IG and FIG algorithms are larger than SA, because in greedy strategy different solution searching sequences could result in different results.

\noindent \textbf{Impact of Power Weight.} As stated before, maximizing the bandwidth utilization is at odds with minimizing the gross energy consumption. Therefore, we introduce a power weight as a tuning parameter to balance between utility maximization and power minimization. Below, we investigate how power weight affects the performance of IG, FIG, and SA, respectively.

Specifically, we studied power weight $\omega$ at five different values: 0.0, 0.5, 1.0, 1.5, and 2.0. The results for power consumption, utility, and energy efficiency are listed in Table \ref{table:weight_power} to \ref{table:weight_efficiency}, respectively\footnote{Note that LA is not listed since it is irrelevant to power weight. Please refer to Table \ref{table_Large} for LA's performance.}. 

For power consumption, we can observe from Table \ref{table:weight_power} that overall, the larger the power weight, the higher the power consumption. For instance, the power consumption of both IG and FIG algorithms decreased by 38\% when $\omega$ increased from 0.0 to 2.0, whereas the power consumption of SA only decreased 12\%. This shows that IG and FIG were more sensitive to power weight variation compared to SA in terms of power consumption. 

\begin{table}[!htbp] \renewcommand{\arraystretch}{1}
\small
\begin{center}
\caption{Power Consumption (/Watt) for Different Power Weights $\omega$} 
\centering 
\scalebox{0.86}{
	\begin{tabular}{|c|c|c|c|c|c|c|c|} 	
	\hline  
	{\backslashbox{$\omega$}{Alg.} }&{IG}&{FIG}&{SA}\\ 
	\hline
	0 &	225.27$\pm$9.64 & 225.08$\pm$11.92 & 238.09$\pm$5.58 \\
	\hline 
	0.5 & 216.54$\pm$11.29 & 220.49$\pm$12.99 & 235.44$\pm$6.39 \\
	\hline 
	1.0 & 213.12$\pm$11.71 & 217.90$\pm$13.33 & 231.94$\pm$7.95 \\
	\hline
	1.5 & 140.57$\pm$6.74 & 140.30$\pm$6.89 &  224.96$\pm$9.61 \\
	\hline
	2.0 & 140.55$\pm$6.77 & 140.34$\pm$6.91 & 209.98$\pm$12.39 \\	
	\hline
	\end{tabular}	}
\label{table:weight_power} 
\end{center}
\end{table}

For utility, as listed in Table \ref{table:weight_utility}, in general, the larger the power weight $\omega$, the smaller the utility. The reason is that the algorithms tend to sacrifice utility to save energy when $\omega$ is large. 

\begin{table}[!htbp] \renewcommand{\arraystretch}{1}
\small
\begin{center}
\caption{Utility for Different Power Weights $\omega$} 
\centering 
\scalebox{0.86}{
	\begin{tabular}{|c|c|c|c|c|c|c|c|} 	
	\hline  
	{\backslashbox{$\omega$}{Alg.} }&{IG}&{FIG}&{SA}\\ 
	\hline
	0 &	3144.60$\pm$149.01 & 3194.70$\pm$174.97 & 3463.70$\pm$75.98 \\
	\hline 
	0.5 & 3154.50$\pm$148.13  & 3186.90$\pm$172.86  & 3466.90$\pm$88.52  \\
	\hline 
	1.0 & 3152.5$\pm$1151.32  & 3171.3$\pm$1171.58  & 3456.80$\pm$184.85  \\
	\hline
	1.5 & 2970.80$\pm$148.58 & 2982.70$\pm$167.06  &  3434.00$\pm$105.80  \\
	\hline
	2.0 & 2977.10$\pm$150.90 & 2986.90$\pm$166.49 & 3408.40$\pm$110.00 \\	
	\hline
	\end{tabular}	}
\label{table:weight_utility} 
\end{center}
\end{table}

For energy efficiency, we can see that the larger the power weight $\omega$, the larger the energy efficiency, as shown in Table \ref{table:weight_efficiency}. In particular, when $\omega$ equals 0, 0.5, and 1.0, IG and FIG algorithms achieved lower energy efficiency than SA; whereas when $\omega$ was above 1.0, IG and FIG achieved even higher energy efficiency compared to SA. Therefore, IG and FIG were more energy efficient than SA when the power weight was higher than 1.0.

\begin{table}[!htbp] \renewcommand{\arraystretch}{1}
\small
\begin{center}
\caption{Energy Efficiency for Different Power Weights $\omega$} 
\centering 
\scalebox{0.86}{
	\begin{tabular}{|c|c|c|c|c|c|c|c|} 	
	\hline  
	{\backslashbox{$\omega$}{Alg.} }&{IG}&{FIG}&{SA}\\ 
	\hline
	0 &	13.98$\pm$0.71 & 14.23$\pm$0.81  & 14.56$\pm$0.45  \\
	\hline 
	0.5 & 14.59$\pm$0.76   & 14.48$\pm$0.80   & 14.73$\pm$0.4  \\
	\hline 
	1.0 & 14.82$\pm$0.84   & 14.59$\pm$0.85   & 14.92$\pm$0.48   \\
	\hline
	1.5 & 21.15$\pm$0.86  & 21.27$\pm$0.95  &  15.28$\pm$0.56   \\
	\hline
	2.0 & 21.20$\pm$0.85 & 21.30$\pm$0.96 & 16.27$\pm$0.83 \\	
	\hline
	\end{tabular}	}
\label{table:weight_efficiency} 
\end{center}
\end{table}

\section{Conclusion}\label{conclution}

In this paper, we studied the problem of improving both bandwidth efficiency and energy efficiency of LTE femtocell networks by taking into consideration interference mitigation and the QoS requirements. The problem is formulated as a joint optimization of UE association and OFDMA scheduling in the framework of a practical MAC protocol. Two distributed algorithms are proposed to approximately solve the problem. Extensive simulation results highlight the considerable improvements achieved by our algorithms compared to the benchmark algorithms in various performance metrics.

\bibliographystyle{abbrv}

\bibliography{reference}

\end{document}